\documentclass[runningheads]{llncs}
\usepackage[T1]{fontenc}
\usepackage{graphicx, hyperref, makecell}
\usepackage{color}

\usepackage{enumitem}
\usepackage{amsmath,amssymb,listings}
\usepackage{tikz}
\usepackage{circuitikz}

\DeclareMathOperator*{\argmin}{arg\,min}

\begin{document}
\title{Dynamic Pricing Algorithms for Online Set Cover}
%
%\titlerunning{Abbreviated paper title}
% If the paper title is too long for the running head, you can set
% an abbreviated paper title here
%
\author{Max Bender\orcidID{0009-0009-4882-7801} \and
Aum Desai\orcidID{0009-0007-1867-1557} \and
Jialin He\orcidID{0009-0006-4044-7191} \and
Oliver Thompson\orcidID{0009-0009-4713-4480} \and
Pramithas Upreti}
\authorrunning{M. Bender et al.}
% First names are abbreviated in the running head.
% If there are more than two authors, 'et al.' is used.
%
\institute{Colby College, Waterville ME, USA}
\maketitle              % typeset the header of the contribution
\begin{abstract} % TODO
We consider dynamic pricing algorithms as applied to the online set cover problem. In the dynamic pricing framework, we assume the standard client server model with the additional constraint that the server can only place prices over the resources they maintain, rather than authoritatively assign them. In response, incoming clients choose the resource which minimizes their disutility when taking into account these additional prices. Our main contributions are the categorization of online algorithms which can be mimicked via dynamic pricing algorithms and the identification of a strongly competitive deterministic algorithm with respect to the frequency parameter of the online set cover input. 

\end{abstract}

\section{Introduction}
In this paper, we'll consider the Online Set Cover problem within the dynamic pricing framework. The framework of dynamic pricing comes from the observation that it is not always in individual clients' best interest to represent themselves honestly in common online client-server problems. Dynamic Pricing addresses this concern by removing the interaction of the client and the server; in dynamic pricing problems, the server maintains prices over its controlled resources and lets the clients choose their desired resource. Typically, clients are assumed to act greedily and select the resource that minimizes the sum of cost due to prices and their own disutility of the resource. As these problems are typically studied in an online setting, the strength of algorithms for these problems are measured via competitive analysis, where we bound the cost of the solution produced by the algorithm / greedy nature of the clients against the solution of the optimal solution in hindsight. 

\subsection{Problem Statement}
In the Online Set Cover (OSC) problem, we are given a universe $X$ and sets $\mathcal{S}\subset\mathcal{P}(X)$, where each set $S$ has associated cost $c_S$. For any element $\eta\in X$ we'll denote the sets that cover $\eta$ by $\mathcal{S}_\eta = \{S\in\mathcal{S}\ |\ \eta\in S\}$. A sequence of elements in $X$ are then requested in an online fashion: for each arriving element, the algorithm must ensure that it is covered by a set in $\mathcal{S}$. At the time of the request of $\eta$, if a set in $\mathcal{S}_\eta$ has already been purchased, the algorithm need not do anything; otherwise, if no set in $\mathcal{S}_\eta$ has been purchased, then the algorithm must choose a set in $\mathcal{S}_\eta$ to irrevocably purchase.

In the Dynamic Pricing Set Cover (DPSC) problem, the input is still the same, but the algorithm instead chooses surcharges for each set prior to the arrival of an element; in particular, an algorithm can update prices in response to the purchase of a set, but not in response to the request of an element. We will denote the surcharge of a set $S$ as $\rho(\mathcal{S})$ and the combined cost of this surcharge and original cost $c_S$ as $\pi(S)$ for each set $S$ (so $\pi(S) = \rho(S) + c_S$). When an element $\eta$ is requested, if it is already covered by a previously purchased set nothing happens; otherwise, the client responsible for the request of $\eta$ purchases the set $$\argmin_{S\in\mathcal{S}_\eta} \pi(S).$$ 

In both problems, the objective is to minimize the sum of the costs of the purchased sets. As these are online problems, we will refer to the competitive ratio of an algorithm $A$ as $$\max_I\frac{A[I]}{\textsc{Opt[I}]},$$ where $I$ ranges over all possible inputs to the online set cover problem, $A[I]$ denotes the cost of the solution produced by $A$ on the input $I$, and $\textsc{Opt}[I]$ denotes the cost of the optimal solution on the input $I$.

\subsection{Related Work}
Offline Set Cover with uniform costs was one of Karp's 21 \textbf{NP}-Hard problems; as a result, numerous results that are of little relevance to this paper exist for this problem. As they are most closely related to our work, we'll discuss results as specified to the online set cover problem, and within the broad framework of Dynamic Pricing. 

\subsubsection{Online Set Cover}
Online Set Cover was first studied in \cite{alon2003}, which found a deterministic $O(\log m \log n)$ competitive primal-dual algorithm, where $m = |\mathcal{S}|$ is the number of sets and $n = |X|$ is the universe size. More recently, \cite{Gupta17} studied Set Cover in the fully dynamic model. In the fully dynamic model, elements can also leave the system (and hence no longer need to be covered) and the algorithm is permitted to a small amount of recourse; typically, this recourse is desired to be quickly computable. In this model, \cite{Gupta17} found a $O(\min(\log n_t, f_t))$ competitive algorithm with $O(1)$ recourse per update (amortized). 

\subsubsection{Dynamic Pricing}
Research into Dynamic Pricing mechanisms arguably began with \cite{cohen2015}, which studied the following classical online problems: 
\begin{itemize}
    \item \textbf{metrical task problems:} \cite{cohen2015} found a strongly competitive algorithm for metrical task systems in trees with competitive factor $O(m)$ (where $m$ is the number of distinct system states).
    \item \textbf{$k$ server problem:} \cite{cohen2015} found a strongly competitive algorithm for the $k$-server problem on line metrics achieving a competitive factor of $k$. This was later improved in \cite{cohen2019} by extending the same result to trees. 
    \item \textbf{minimal metric matching:} \cite{cohen2015} found a $O(\log \Delta)$-competitive algorithm for metrical matching on the line where $\Delta$ is the aspect ratio of the space. This can be improved to $O(\log n)$ when given a reasonably accurate estimate of the cost of the optimal solution ahead of time. This was later improved in \cite{arndt2023} by developing a $O(\log n)$ randomized competitive algorithm on the line without an estimate of the cost of the optimal solution. Additionally, \cite{bender2020} characterized precisely which algorithms for the online minimal metric matching problem for trees could be mimicked by dynamic pricing algorithms and developed a polylog competitive algorithm for metric search on trees. This was later followed up with a polylog competitive matching algorithm on spiders in \cite{bender2021}.
\end{itemize}
There have also been applications of the dynamic pricing framework to job scheduling problems. In particular, in \cite{feldman2017} they worked on job scheduling in a variety of contexts under the makespan minimization objective, the results of which are summarized in the following table: 
\begin{center}
    \begin{tabular}{@{\hspace{1em}}c@{\hspace{1em}}|@{\hspace{1em}}c@{\hspace{1em}}|@{\hspace{1em}}c@{\hspace{1em}}}
        Machine Model & Static Pricing & Dynamic Pricing\\\hline
        Identical & $O(1)$ & $O(1)$ \\
        Related & $\Theta(\log m)$ & $O(1)$\\
        Restricted & $\Theta(\log m)$ & $\Theta(\log m)$\\
        Unrelated & $\Theta(m)$ & $\Theta(m)$
    \end{tabular}
\end{center}
Of particular interest in this table is the lowerbound of $\Omega(m)$ for unrelated machines under dynamic pricing, as there exists a $O(\log m)$ competitive online algorithm for unrelated machines. To our knowledge, this is the first and only known problem which is asymptotically harder under the dynamic pricing framework. Additionally, in \cite{im2017} developed an immediate-dispatch $O(1)$-competitive algorithm for maximum flow on related machines.

Similar to the dynamic pricing framework, \cite{eden2018} considered a framework of prompt scheduling of jobs while minimizing sum of completion times. In doing so, their algorithm would create `menus' with dyanmic prices associated with intervals for clients to purchase for running their jobs. Their results are summarized in the following table: 
\begin{center}
    \begin{tabular}{@{\hspace{0.5em}}c@{\hspace{0.5em}}|@{\hspace{0.5em}}c@{\hspace{0.5em}}|@{\hspace{0.5em}}c@{\hspace{0.5em}}|@{\hspace{0.5em}}c@{\hspace{0.5em}}|@{\hspace{0.5em}}c}
        \makecell{Processing\\Time} & \makecell{Job\\Weight} & \makecell{Menu\\Entries} & \makecell{Upper Bound\\(Deterministic)} & \makecell{Lower Bound\\(Randomized)}\\\hline
        $p_j\in\mathbb{Z}^+$ & $w_j = 1$ & \makecell{intervals\\no prices} & $O(\log P_{\text{max}})$ & $\Omega(\log P_{\text{max}})$\\\hline
        $p_j = 1$ & $w_j\in\mathbb{Z}^+$ & \makecell{intervals\\with prices} & \makecell{$O(\log W_\text{max}\times(\log n$\\$+\log\log W_\text{max})$} & $\Omega(\log W_\text{max})$\\\hline
        $p_j\in\mathbb{Z}^+$ & $w_j\in\mathbb{Z}^+$ & \makecell{intervals\\with prices} & \makecell{$O((\log n+\log P_\text{max})$\\$\times\log B_\text{max})$} & \makecell{$\Omega(\max (\log B_\text{max},$\\$\log P_\text{max}))$}
    \end{tabular}
\end{center}
In the above table, $P_\text{max}$ refers to the maximal processing time, $W_\text{max}$ refers to the highest weighted job, and $B_\text{max}$ refers to an apriori upperbound on the cost of the optimal solution.

\subsection{Our Results}
% TODO: Basically a description of the flow of our paper
Our primary result is the following theorem:
\begin{theorem}
    There is a $\Theta(f)$-competitive dynamic pricing algorithm for the online set cover problem and this is optimal for deterministic algorithms.
\end{theorem} In the above theorem, $f = \max_\eta|\mathcal{S}_\eta|$ is the \textbf{frequency} of the input.

We begin by motivating the work in section \ref{sec:Greedy} by analyzing the natural \textsc{Greedy} algorithm and observing that its competitive ratio is unbounded in the parameters of $n$, $m$, and $f$. In section \ref{sec:Monotonicty}, we characterize precisely which algorithms for the OSC problem can be exactly mimicked by dynamic pricing algorithms and develop an algorithm for computing pricing schemes that mimics priceable algorithms. Finally, in section \ref{sec:PDAlg} we review an algorithm from \cite{buchbinder2009} and show that it is monotone and therefore priceable. Additionally, we show that under the frequency parameter $f$ it is optimal with respect to the competitive ratio for deterministic algorithms.

\section{Performance of \textsc{Greedy}}\label{sec:Greedy}
We first show that $\textsc{Greedy}$ algorithm performs badly on the OSC problem. In particular, we consider the algorithm in which the cheapest set available for uncovered elements is chosen; equivalently, this is the solution produced by agents under a dynamic pricing scheme where the algorithm selects a uniformly 0 surcharge function.

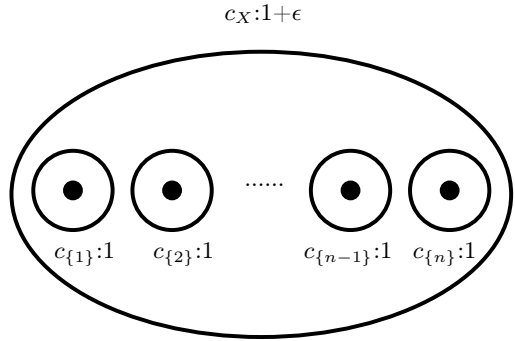
\begin{figure}[!h]
\centering
\tikzset{every picture/.style={line width=0.75pt}} %set default line width to 0.75pt        

\begin{tikzpicture}[x=0.75pt,y=0.75pt,yscale=-1,xscale=1]
%uncomment if require: \path (0,300); %set diagram left start at 0, and has height of 300

%Shape: Circle [id:dp36825642024681415] 
\draw  [line width=1.5]  (205,130) .. controls (205,118.95) and (213.95,110) .. (225,110) .. controls (236.05,110) and (245,118.95) .. (245,130) .. controls (245,141.05) and (236.05,150) .. (225,150) .. controls (213.95,150) and (205,141.05) .. (205,130) -- cycle ;
%Shape: Ellipse [id:dp9143593911304182] 
\draw  [line width=1.5]  (194,132) .. controls (194,92.24) and (250.41,60) .. (320,60) .. controls (389.59,60) and (446,92.24) .. (446,132) .. controls (446,171.76) and (389.59,204) .. (320,204) .. controls (250.41,204) and (194,171.76) .. (194,132) -- cycle ;
%Shape: Circle [id:dp11513631880217745] 
\draw  [fill={rgb, 255:red, 0; green, 0; blue, 0 }  ,fill opacity=1 ] (220.5,130) .. controls (220.5,127.51) and (222.51,125.5) .. (225,125.5) .. controls (227.49,125.5) and (229.5,127.51) .. (229.5,130) .. controls (229.5,132.49) and (227.49,134.5) .. (225,134.5) .. controls (222.51,134.5) and (220.5,132.49) .. (220.5,130) -- cycle ;
%Shape: Circle [id:dp9605008962748878] 
\draw  [line width=1.5]  (255,130) .. controls (255,118.95) and (263.95,110) .. (275,110) .. controls (286.05,110) and (295,118.95) .. (295,130) .. controls (295,141.05) and (286.05,150) .. (275,150) .. controls (263.95,150) and (255,141.05) .. (255,130) -- cycle ;
%Shape: Circle [id:dp39435258831671427] 
\draw  [fill={rgb, 255:red, 0; green, 0; blue, 0 }  ,fill opacity=1 ] (270.5,130) .. controls (270.5,127.51) and (272.51,125.5) .. (275,125.5) .. controls (277.49,125.5) and (279.5,127.51) .. (279.5,130) .. controls (279.5,132.49) and (277.49,134.5) .. (275,134.5) .. controls (272.51,134.5) and (270.5,132.49) .. (270.5,130) -- cycle ;
%Shape: Circle [id:dp8651877611918519] 
\draw  [line width=1.5]  (345,130) .. controls (345,118.95) and (353.95,110) .. (365,110) .. controls (376.05,110) and (385,118.95) .. (385,130) .. controls (385,141.05) and (376.05,150) .. (365,150) .. controls (353.95,150) and (345,141.05) .. (345,130) -- cycle ;
%Shape: Circle [id:dp6460350145575198] 
\draw  [fill={rgb, 255:red, 0; green, 0; blue, 0 }  ,fill opacity=1 ] (360.5,130) .. controls (360.5,127.51) and (362.51,125.5) .. (365,125.5) .. controls (367.49,125.5) and (369.5,127.51) .. (369.5,130) .. controls (369.5,132.49) and (367.49,134.5) .. (365,134.5) .. controls (362.51,134.5) and (360.5,132.49) .. (360.5,130) -- cycle ;
%Shape: Circle [id:dp24919286735049884] 
\draw  [line width=1.5]  (395,130) .. controls (395,118.95) and (403.95,110) .. (415,110) .. controls (426.05,110) and (435,118.95) .. (435,130) .. controls (435,141.05) and (426.05,150) .. (415,150) .. controls (403.95,150) and (395,141.05) .. (395,130) -- cycle ;
%Shape: Circle [id:dp11799043631847983] 
\draw  [fill={rgb, 255:red, 0; green, 0; blue, 0 }  ,fill opacity=1 ] (410.5,130) .. controls (410.5,127.51) and (412.51,125.5) .. (415,125.5) .. controls (417.49,125.5) and (419.5,127.51) .. (419.5,130) .. controls (419.5,132.49) and (417.49,134.5) .. (415,134.5) .. controls (412.51,134.5) and (410.5,132.49) .. (410.5,130) -- cycle ;

% Text Node
\draw (310,125) node [anchor=north west][inner sep=0.75pt]   [align=left] {......};
% Text Node
\draw (300,35) node [anchor=north west][inner sep=0.75pt]   [align=left] {{\small $c_{X}$:1+$\epsilon$}};
% Text Node
\draw (214,155) node [anchor=north west][inner sep=0.75pt]   [align=left] {{\small $c_{\{1\}}$:1}};
% Text Node
\draw (264,155) node [anchor=north west][inner sep=0.75pt]   [align=left] {{\small $c_{\{2\}}$:1}};
% Text Node
\draw (340,155) node [anchor=north west][inner sep=0.75pt]   [align=left] {{\small $c_{\{n-1\}}$:1}};
% Text Node
\draw (395,155) node [anchor=north west][inner sep=0.75pt]   [align=left] {{\small $c_{\{n\}}$:1}};
\end{tikzpicture}
\caption{A hard instance for \textsc{Greedy}}
\end{figure}

\begin{lemma}
    The \textsc{Greedy} algorithm has an unbounded competitive ratio with respect to the frequency parameter.
\end{lemma}

% \begin{definition}[Frequency]
% Given an input $(X, \mathcal{S})$ to the OSC/DPSC problem, the \textbf{frequency} of the input is defined as $\max_\eta |\mathcal{S}_\eta|$.
% \end{definition}

\begin{proof}
    Consider elements $X = \{1,2,3,...,n-1,n\}$ that each belong to a singleton set of cost 1 while the entire collection $X$ has cost $1 + \epsilon$. The greedy algorithm would assign each element to its singleton, giving a total cost of $n$, while the optimal solution is $1 + \epsilon$. If we increase $n$, the frequency of the instance remains at 2 while the cost of \textsc{Greedy} increases to $n$ and \textsc{Opt} remains 1 + $\epsilon$. Thus \textsc{Greedy} has an unbounded competitive ratio with respect to frequency. 
\end{proof}

Notably, this also induces a linear lowerbound for the \textsc{Greedy} algorithm, showing we'd need to be at least reasonably clever to obtain a $O(\log n\log  m)$ pricing scheme to match the best known algorithm for the OSC under these parameters.

\section{Equivalence of Monotonicity and Priceability}\label{sec:Monotonicty}
Knowing the poor performance of the greedy approach, pricing schemes could be introduced to improve upon this performance. Since there are many algorithms for the OSC problem already, it is a natural step to try and determine which of them are priceable. Specifically, for which algorithms does there exist a pricing scheme that will replicate the behavior of the algorithm when each element chooses the least expensive set that it is covered by? We call this strategy \textbf{mimickry}. In this section, we introduce a categorization of priceability for the OSC problem and introduce an algorithm $\textsc{PathPrice}$ that creates a pricing scheme for any priceable OSC algorithm. This pricing scheme will \textit{mimic} the original OSC algorithm; that is, each incoming client will choose the set to cover itself that the algorithm would have assigned to it. 

In many of the coming definitions we will make reference to assignment \textit{schemes} and pricing \textit{schemes}: in such scenarios, we are talking about a fixed moment in time directly prior to the arrival of a request for a particular element. An assignment scheme is induced by an algorithm for the OSC: it is the function $A: X\rightarrow\mathcal{S}$ where $A(\eta)$ is the set the algorithm would choose to cover $\eta$ should it be the next element requested. We will use the notation $\eta\rightarrow_\mathcal{A} T$ to denote $A(\eta) = T$.

\begin{definition}[Preference Graph]
Given an assignment scheme $\mathcal{A}$, the associated preference graph is the digraph $G = (V, E)$ defined by $V = \mathcal{S}$, where $S_{v}$ is the set corresponding to vertex $v \in G$ and $E = \{(S, T)\ |\ S,T \in \mathcal{S}\ |\ \exists \eta\in S\cap T: \eta\rightarrow_\mathcal{A} T\}$
\end{definition}

Preference graphs are simply the graphical interpretation of how choices are made according to a given assignment scheme. As we shall see, this choices are the defining aspects of an assignment scheme which allow us to determine prices. 

\begin{definition}[Monotonicity]
An algorithm for the Online Set Cover problem is \textbf{monotone} if every assignment scheme $\mathcal{A}$ it determines induces an acyclic preference graph.
\end{definition}

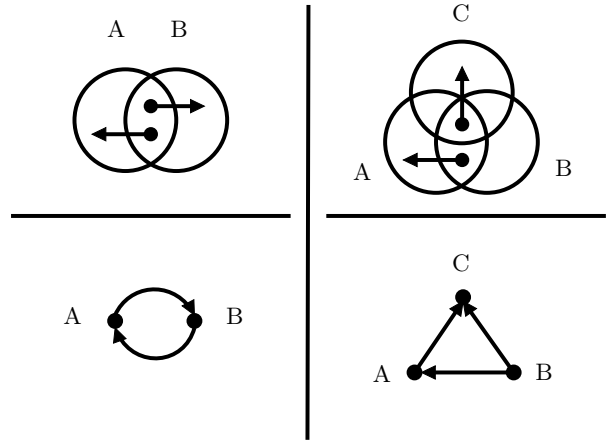
\begin{figure}[!h]
\centering
\tikzset{every picture/.style={line width=0.75pt}} %set default line width to 0.75pt        
\begin{tikzpicture}[x=0.75pt,y=0.75pt,yscale=-1,xscale=1]
%uncomment if require: \path (0,300); %set diagram left start at 0, and has height of 300

%Straight Lines [id:da05156662332656947] 
\draw [line width=1.5]    (172,143) -- (313,143) ;
%Straight Lines [id:da9258027652809535] 
\draw [line width=1.5]    (322,36.5) -- (321.5,256) ;
%Shape: Ellipse [id:dp6380322319272236] 
\draw  [line width=1.5]  (204,94.67) .. controls (204,80.49) and (215.49,69) .. (229.67,69) .. controls (243.84,69) and (255.33,80.49) .. (255.33,94.67) .. controls (255.33,108.84) and (243.84,120.33) .. (229.67,120.33) .. controls (215.49,120.33) and (204,108.84) .. (204,94.67) -- cycle ;
%Shape: Circle [id:dp38694526096471527] 
\draw  [line width=1.5]  (229.67,94.67) .. controls (229.67,80.49) and (241.16,69) .. (255.33,69) .. controls (269.51,69) and (281,80.49) .. (281,94.67) .. controls (281,108.84) and (269.51,120.33) .. (255.33,120.33) .. controls (241.16,120.33) and (229.67,108.84) .. (229.67,94.67) -- cycle ;
%Shape: Circle [id:dp969714302665134] 
\draw  [fill={rgb, 255:red, 0; green, 0; blue, 0 }  ,fill opacity=1 ] (245.6,87.59) .. controls (245.6,85.88) and (244.21,84.49) .. (242.5,84.49) .. controls (240.79,84.49) and (239.4,85.88) .. (239.4,87.59) .. controls (239.4,89.3) and (240.79,90.68) .. (242.5,90.68) .. controls (244.21,90.68) and (245.6,89.3) .. (245.6,87.59) -- cycle ;
%Shape: Circle [id:dp7817318748613897] 
\draw  [fill={rgb, 255:red, 0; green, 0; blue, 0 }  ,fill opacity=1 ] (245.6,101.75) .. controls (245.6,100.04) and (244.21,98.65) .. (242.5,98.65) .. controls (240.79,98.65) and (239.4,100.04) .. (239.4,101.75) .. controls (239.4,103.46) and (240.79,104.84) .. (242.5,104.84) .. controls (244.21,104.84) and (245.6,103.46) .. (245.6,101.75) -- cycle ;
%Straight Lines [id:da5090672659543536] 
\draw [line width=1.5]    (242.5,87.59) -- (265.49,87.59) ;
\draw [shift={(269.49,87.59)}, rotate = 180] [fill={rgb, 255:red, 0; green, 0; blue, 0 }  ][line width=0.08]  [draw opacity=0] (6.97,-3.35) -- (0,0) -- (6.97,3.35) -- cycle    ;
%Straight Lines [id:da06750455692169699] 
\draw [line width=1.5]    (242.5,101.75) -- (216.85,101.75) ;
\draw [shift={(212.85,101.75)}, rotate = 360] [fill={rgb, 255:red, 0; green, 0; blue, 0 }  ][line width=0.08]  [draw opacity=0] (6.97,-3.35) -- (0,0) -- (6.97,3.35) -- cycle    ;
%Shape: Circle [id:dp015863675029447633] 
\draw  [fill={rgb, 255:red, 0; green, 0; blue, 0 }  ,fill opacity=1 ] (228,196) .. controls (228,194.07) and (226.43,192.5) .. (224.5,192.5) .. controls (222.57,192.5) and (221,194.07) .. (221,196) .. controls (221,197.93) and (222.57,199.5) .. (224.5,199.5) .. controls (226.43,199.5) and (228,197.93) .. (228,196) -- cycle ;
%Shape: Circle [id:dp2906673209038604] 
\draw  [fill={rgb, 255:red, 0; green, 0; blue, 0 }  ,fill opacity=1 ] (268,196) .. controls (268,194.07) and (266.43,192.5) .. (264.5,192.5) .. controls (262.57,192.5) and (261,194.07) .. (261,196) .. controls (261,197.93) and (262.57,199.5) .. (264.5,199.5) .. controls (266.43,199.5) and (268,197.93) .. (268,196) -- cycle ;
%Curve Lines [id:da5551146741910116] 
\draw [line width=1.5]    (224.5,193) .. controls (231.51,177.57) and (252.08,175.26) .. (262.48,189.74) ;
\draw [shift={(264.5,193)}, rotate = 242.18] [fill={rgb, 255:red, 0; green, 0; blue, 0 }  ][line width=0.08]  [draw opacity=0] (6.97,-3.35) -- (0,0) -- (6.97,3.35) -- cycle    ;
%Curve Lines [id:da15157638976896526] 
\draw [line width=1.5]    (264.5,199.5) .. controls (259.33,218.3) and (235.16,220.31) .. (226.05,203.04) ;
\draw [shift={(224.5,199.5)}, rotate = 70.35] [fill={rgb, 255:red, 0; green, 0; blue, 0 }  ][line width=0.08]  [draw opacity=0] (6.97,-3.35) -- (0,0) -- (6.97,3.35) -- cycle    ;
%Shape: Circle [id:dp2679843922099123] 
\draw  [line width=1.5]  (360.67,105.67) .. controls (360.67,91.49) and (372.16,80) .. (386.33,80) .. controls (400.51,80) and (412,91.49) .. (412,105.67) .. controls (412,119.84) and (400.51,131.33) .. (386.33,131.33) .. controls (372.16,131.33) and (360.67,119.84) .. (360.67,105.67) -- cycle ;
%Shape: Circle [id:dp8909390425275243] 
\draw  [line width=1.5]  (386.33,105.67) .. controls (386.33,91.49) and (397.82,80) .. (412,80) .. controls (426.18,80) and (437.67,91.49) .. (437.67,105.67) .. controls (437.67,119.84) and (426.18,131.33) .. (412,131.33) .. controls (397.82,131.33) and (386.33,119.84) .. (386.33,105.67) -- cycle ;
%Shape: Circle [id:dp8758541560615256] 
\draw  [line width=1.5]  (373.67,80.67) .. controls (373.67,66.49) and (385.16,55) .. (399.33,55) .. controls (413.51,55) and (425,66.49) .. (425,80.67) .. controls (425,94.84) and (413.51,106.33) .. (399.33,106.33) .. controls (385.16,106.33) and (373.67,94.84) .. (373.67,80.67) -- cycle ;
%Shape: Circle [id:dp854456431299335] 
\draw  [fill={rgb, 255:red, 0; green, 0; blue, 0 }  ,fill opacity=1 ] (402.6,114.75) .. controls (402.6,113.04) and (401.21,111.65) .. (399.5,111.65) .. controls (397.79,111.65) and (396.4,113.04) .. (396.4,114.75) .. controls (396.4,116.46) and (397.79,117.84) .. (399.5,117.84) .. controls (401.21,117.84) and (402.6,116.46) .. (402.6,114.75) -- cycle ;
%Shape: Circle [id:dp3953776792121668] 
\draw  [fill={rgb, 255:red, 0; green, 0; blue, 0 }  ,fill opacity=1 ] (402.6,96.75) .. controls (402.6,95.04) and (401.21,93.65) .. (399.5,93.65) .. controls (397.79,93.65) and (396.4,95.04) .. (396.4,96.75) .. controls (396.4,98.46) and (397.79,99.84) .. (399.5,99.84) .. controls (401.21,99.84) and (402.6,98.46) .. (402.6,96.75) -- cycle ;
%Straight Lines [id:da3785069176757454] 
\draw [line width=1.5]    (399.5,114.75) -- (373.85,114.75) ;
\draw [shift={(369.85,114.75)}, rotate = 360] [fill={rgb, 255:red, 0; green, 0; blue, 0 }  ][line width=0.08]  [draw opacity=0] (6.97,-3.35) -- (0,0) -- (6.97,3.35) -- cycle    ;
%Straight Lines [id:da19606773766072805] 
\draw [line width=1.5]    (399.5,96.75) -- (399.5,71.67) ;
\draw [shift={(399.5,67.67)}, rotate = 90] [fill={rgb, 255:red, 0; green, 0; blue, 0 }  ][line width=0.08]  [draw opacity=0] (6.97,-3.35) -- (0,0) -- (6.97,3.35) -- cycle    ;
%Shape: Circle [id:dp44775995742746555] 
\draw  [fill={rgb, 255:red, 0; green, 0; blue, 0 }  ,fill opacity=1 ] (379,222) .. controls (379,220.07) and (377.43,218.5) .. (375.5,218.5) .. controls (373.57,218.5) and (372,220.07) .. (372,222) .. controls (372,223.93) and (373.57,225.5) .. (375.5,225.5) .. controls (377.43,225.5) and (379,223.93) .. (379,222) -- cycle ;
%Shape: Circle [id:dp1453366337647175] 
\draw  [fill={rgb, 255:red, 0; green, 0; blue, 0 }  ,fill opacity=1 ] (429,222) .. controls (429,220.07) and (427.43,218.5) .. (425.5,218.5) .. controls (423.57,218.5) and (422,220.07) .. (422,222) .. controls (422,223.93) and (423.57,225.5) .. (425.5,225.5) .. controls (427.43,225.5) and (429,223.93) .. (429,222) -- cycle ;
%Shape: Circle [id:dp1960455758777342] 
\draw  [fill={rgb, 255:red, 0; green, 0; blue, 0 }  ,fill opacity=1 ] (403.5,184) .. controls (403.5,182.07) and (401.93,180.5) .. (400,180.5) .. controls (398.07,180.5) and (396.5,182.07) .. (396.5,184) .. controls (396.5,185.93) and (398.07,187.5) .. (400,187.5) .. controls (401.93,187.5) and (403.5,185.93) .. (403.5,184) -- cycle ;
%Straight Lines [id:da7797243688282236] 
\draw [line width=1.5]    (375.5,222) -- (397.75,189.31) ;
\draw [shift={(400,186)}, rotate = 124.24] [fill={rgb, 255:red, 0; green, 0; blue, 0 }  ][line width=0.08]  [draw opacity=0] (6.97,-3.35) -- (0,0) -- (6.97,3.35) -- cycle    ;
%Straight Lines [id:da8461127310504171] 
\draw [line width=1.5]    (425.5,222) -- (402.31,189.26) ;
\draw [shift={(400,186)}, rotate = 54.69] [fill={rgb, 255:red, 0; green, 0; blue, 0 }  ][line width=0.08]  [draw opacity=0] (6.97,-3.35) -- (0,0) -- (6.97,3.35) -- cycle    ;
%Straight Lines [id:da7213442208422369] 
\draw [line width=1.5]    (425.5,222) -- (383,222) ;
\draw [shift={(379,222)}, rotate = 360] [fill={rgb, 255:red, 0; green, 0; blue, 0 }  ][line width=0.08]  [draw opacity=0] (6.97,-3.35) -- (0,0) -- (6.97,3.35) -- cycle    ;
%Straight Lines [id:da561454560043759] 
\draw [line width=1.5]    (331,143) -- (472,143) ;

% Text Node
\draw (219,43) node [anchor=north west][inner sep=0.75pt]   [align=left] {A};
% Text Node
\draw (251,43) node [anchor=north west][inner sep=0.75pt]   [align=left] {B};
% Text Node
\draw (197,188) node [anchor=north west][inner sep=0.75pt]   [align=left] {A};
% Text Node
\draw (279,188) node [anchor=north west][inner sep=0.75pt]   [align=left] {B};
% Text Node
\draw (343,115) node [anchor=north west][inner sep=0.75pt]   [align=left] {A};
% Text Node
\draw (445,113) node [anchor=north west][inner sep=0.75pt]   [align=left] {B};
% Text Node
\draw (393,35) node [anchor=north west][inner sep=0.75pt]   [align=left] {C};
% Text Node
\draw (353,217) node [anchor=north west][inner sep=0.75pt]   [align=left] {A};
% Text Node
\draw (435,216) node [anchor=north west][inner sep=0.75pt]   [align=left] {B};
% Text Node
\draw (393,161) node [anchor=north west][inner sep=0.75pt]   [align=left] {C};
\end{tikzpicture}
\caption{A non-monotone and a monotone instance. Top left: An example assignment scheme where arrows indicate the set that is assigned to cover the element if it is introduced next. Bottom left: the resulting cyclic preference graph from the assignment scheme in the top left. This is an unpriceable edge relation. Right side: a valid assignment scheme and its resulting preference graph.}
\end{figure}

\begin{lemma}
    Every pricing scheme induces a monotone assignment scheme.
\end{lemma}

\begin{proof}
Let $\pi$ be a pricing scheme and let $\mathcal{A}$ be the associated assignment scheme with preference graph $G$. Suppose that  $\mathcal{A}$ is not monotone. Then, there must exist a cycle $C$ in $G$. The vertices in the cycle $C$ correspond to the sequence of sets $(S_1, ..., S_k)$ such that for each $1\leq i\leq k$ there is an $\eta_i\in S_i\cap S_{i+1}$\footnote{Letting $S_{n+1} = S_1$} such that $\eta_i\rightarrow_\mathcal{A} S_{i+1}$. This would mean that $\pi(S_i) < \pi(S_{i+1})$ for all $1\leq i\leq k$. However, this is a contradiction, because we have $\pi(S_k)<\pi(S_1) < \pi(S_2) < ... < \pi(S_{k-1}) < \pi(S_k)$ which is an unpriceable set of inequalities. Thus the pricing scheme must be monotone.
\end{proof}

% TODO THis shouldn't be a definition, should introduce in paragraph form

A monotone preference graph G is a directed acyclic graph where an edge from vertex $v_k$ to $v_j$ implies that $\pi(S_k) > \pi(S_j)$ in the final pricing scheme in order to replicate the behavior of an assignment scheme $\mathcal{A}$. So the algorithm that assigns valid prices to each set must satisfy the above inequality for each edge in the preference graph. The obvious way to do this is price the vertices according to the longest path to which they belong first, which is the idea behind $\textsc{PathPrice}$. The algorithm $\textsc{PathPrice}$ returns a surcharge function $\rho(S): \mathcal{S} \rightarrow \mathbb{R}^+$ such that $\pi(S_k) = \rho(S_k) + c_{S_k} > \rho(S_j) + c_{S_j} = \pi(S_j)$ for all $k \rightarrow j \in G$. In the process of building the $\rho$, we will create a length function $l(v): v(G) \longrightarrow (NaN \cup \mathbb{R}^+$), so the $l(v)$ is either not yet set $(NaN)$ or equal to the longest path's number of vertices stemming from the vertex $v$. 

%\begin{definition}[Surcharge Function $\rho(S)$]\\
%$\rho(S): \mathcal{S} \rightarrow (\{NaN\} \bigcup \mathbb{R}^+$) is the surcharge function that assigns a surcharge of either NaN or some non-negative real number to each set in S. For the purposes of calculating $\pi(S)$ and $\textsc{NextPath}$, if $\rho(S) = NaN$, temporarily set $\rho(S) = 0$. If $\rho(S_v) = NaN$, $\textsc{NextPath}$ will evaluate it as 0 for calculating path length
%\end{definition}

\begin{definition}[\textsc{NextPath}]\\
    \textbf{Input}: A digraph $G$ and a length function $l(v)$\\
    \textbf{Output}:The associated vertex sequence of the  longest path in a graph defined as: ($\#$ of vertices in the path + $l(v)$), where path length is at least 1 and $v$ is the last vertex in the path. The vertex sequence is returned as a list of vertices where $v[i]$ has an edge to $v[i+1]$. If $l(v) = NaN$, $\textsc{NextPath}$ will evaluate it as 0 for the purpose of calculating path length. 
\end{definition}

We will use \textsc{NextPath} as a subroutine to iteratively process the graph according to the longest paths, constructing the surcharge function along the way. 

\begin{definition}[\textsc{PathPrice}]\end{definition}
% TODO Low priority picture for this
\begin{lstlisting}[mathescape=true, numbers=left]
Given a monotone assignment scheme: 
    Create a Preference Graph $G$.  

Initialize $l(v)$ to $NaN$ for all vertices in G
Let $C_{max}$ be the maximum cost of $S \in \mathcal{S}$.

While the number of edges in $G$ is not zero:
    v = $\textsc{NextPath}(G,l)$
    k = v.length
    For $i = k - 1, ..., 1$:
        If $l(v_k)$ is $NaN$,       
            set $l(v_k)$ = 0.
        If $l(v_i)$ is not $NaN$, continue.
        Else, set $l(v_i)= l(v_{i+1})+1$.
    Delete edges $v_i,v_{i+1}$ for all vertices of v in G
    
For each vertex $u$ in $G$:
    $\rho(S_u) \leftarrow l(u) + (C_{max} - c_{S_u})$
Output: $\rho: \mathcal{S}\rightarrow\mathbb{R^+}$.
\end{lstlisting}

The following lemmas show that a valid pricing scheme exists for any monotone OSC algorithm by creating a valid $\rho$ function for every assignment scheme $\mathcal{A}$  using $\textsc{PathPrice}$. 

\begin{figure}
    \centering

\tikzset{every picture/.style={line width=0.75pt}} %set default line width to 0.75pt        

\begin{tikzpicture}[x=0.75pt,y=0.75pt,yscale=-1,xscale=1]
%uncomment if require: \path (0,300); %set diagram left start at 0, and has height of 300

%Rounded Rect [id:dp37524946892433286] 
\draw  [draw opacity=0][fill={rgb, 255:red, 176; green, 213; blue, 223 }  ,fill opacity=1 ] (557.82,147.5) .. controls (559.71,147.5) and (561.25,149.04) .. (561.25,150.93) -- (561.25,159.07) .. controls (561.25,160.96) and (559.71,162.5) .. (557.82,162.5) -- (458.18,162.5) .. controls (456.29,162.5) and (454.75,160.96) .. (454.75,159.07) -- (454.75,150.93) .. controls (454.75,149.04) and (456.29,147.5) .. (458.18,147.5) -- cycle ;
%Rounded Rect [id:dp3811001716227993] 
\draw  [draw opacity=0][fill={rgb, 255:red, 176; green, 213; blue, 223 }  ,fill opacity=1 ] (323.45,100) .. controls (325.41,100) and (327,101.59) .. (327,103.55) -- (327,111.95) .. controls (327,113.91) and (325.41,115.5) .. (323.45,115.5) -- (273.45,115.5) .. controls (271.49,115.5) and (269.9,113.91) .. (269.9,111.95) -- (269.9,103.55) .. controls (269.9,101.59) and (271.49,100) .. (273.45,100) -- cycle ;
%Rounded Rect [id:dp5478743375253592] 
\draw  [draw opacity=0][fill={rgb, 255:red, 176; green, 213; blue, 223 }  ,fill opacity=1 ] (266.4,103.52) .. controls (266.4,101.58) and (267.98,100) .. (269.92,100) -- (280.48,100) .. controls (282.42,100) and (284,101.58) .. (284,103.52) -- (284,204.48) .. controls (284,206.42) and (282.42,208) .. (280.48,208) -- (269.92,208) .. controls (267.98,208) and (266.4,206.42) .. (266.4,204.48) -- cycle ;
%Rounded Rect [id:dp8394364180868337] 
\draw  [draw opacity=0][fill={rgb, 255:red, 176; green, 213; blue, 223 }  ,fill opacity=1 ] (370.47,194.13) .. controls (372.37,194.13) and (373.9,195.67) .. (373.9,197.57) -- (373.9,205.7) .. controls (373.9,207.6) and (372.37,209.13) .. (370.47,209.13) -- (270.84,209.13) .. controls (268.94,209.13) and (267.4,207.6) .. (267.4,205.7) -- (267.4,197.57) .. controls (267.4,195.67) and (268.94,194.13) .. (270.84,194.13) -- cycle ;
%Straight Lines [id:da9524599035246084] 
\draw [line width=1.5]    (218,85) -- (218,230.5) ;
%Straight Lines [id:da13265115214335954] 
\draw [line width=1.5]    (138,107) -- (99,107) ;
\draw [shift={(95,107)}, rotate = 360] [fill={rgb, 255:red, 0; green, 0; blue, 0 }  ][line width=0.08]  [draw opacity=0] (6.97,-3.35) -- (0,0) -- (6.97,3.35) -- cycle    ;
%Straight Lines [id:da8066559182386415] 
\draw [line width=1.5]    (95,111) -- (95,147) ;
\draw [shift={(95,151)}, rotate = 270] [fill={rgb, 255:red, 0; green, 0; blue, 0 }  ][line width=0.08]  [draw opacity=0] (6.97,-3.35) -- (0,0) -- (6.97,3.35) -- cycle    ;
%Straight Lines [id:da6801728304932313] 
\draw [line width=1.5]    (95,159) -- (95,195) ;
\draw [shift={(95,199)}, rotate = 270] [fill={rgb, 255:red, 0; green, 0; blue, 0 }  ][line width=0.08]  [draw opacity=0] (6.97,-3.35) -- (0,0) -- (6.97,3.35) -- cycle    ;
%Straight Lines [id:da9962256153722471] 
\draw [line width=1.5]    (95,201) -- (131,201) ;
\draw [shift={(135,201)}, rotate = 180] [fill={rgb, 255:red, 0; green, 0; blue, 0 }  ][line width=0.08]  [draw opacity=0] (6.97,-3.35) -- (0,0) -- (6.97,3.35) -- cycle    ;
%Straight Lines [id:da8391845093615999] 
\draw [line width=1.5]    (141,201) -- (177,201) ;
\draw [shift={(181,201)}, rotate = 180] [fill={rgb, 255:red, 0; green, 0; blue, 0 }  ][line width=0.08]  [draw opacity=0] (6.97,-3.35) -- (0,0) -- (6.97,3.35) -- cycle    ;
%Straight Lines [id:da21866878325694672] 
\draw [line width=1.5]    (94,154) -- (130,154) ;
\draw [shift={(134,154)}, rotate = 180] [fill={rgb, 255:red, 0; green, 0; blue, 0 }  ][line width=0.08]  [draw opacity=0] (6.97,-3.35) -- (0,0) -- (6.97,3.35) -- cycle    ;
%Straight Lines [id:da24860592921439473] 
\draw [line width=1.5]    (140,154) -- (176,154) ;
\draw [shift={(180,154)}, rotate = 180] [fill={rgb, 255:red, 0; green, 0; blue, 0 }  ][line width=0.08]  [draw opacity=0] (6.97,-3.35) -- (0,0) -- (6.97,3.35) -- cycle    ;
%Shape: Circle [id:dp5971936049187383] 
\draw  [fill={rgb, 255:red, 0; green, 0; blue, 0 }  ,fill opacity=1 ] (98.1,107.1) .. controls (98.1,105.39) and (96.71,104) .. (95,104) .. controls (93.29,104) and (91.9,105.39) .. (91.9,107.1) .. controls (91.9,108.81) and (93.29,110.2) .. (95,110.2) .. controls (96.71,110.2) and (98.1,108.81) .. (98.1,107.1) -- cycle ;
%Shape: Circle [id:dp9084872326774238] 
\draw  [fill={rgb, 255:red, 0; green, 0; blue, 0 }  ,fill opacity=1 ] (98.1,154.1) .. controls (98.1,152.39) and (96.71,151) .. (95,151) .. controls (93.29,151) and (91.9,152.39) .. (91.9,154.1) .. controls (91.9,155.81) and (93.29,157.2) .. (95,157.2) .. controls (96.71,157.2) and (98.1,155.81) .. (98.1,154.1) -- cycle ;
%Shape: Circle [id:dp33288061546075043] 
\draw  [fill={rgb, 255:red, 0; green, 0; blue, 0 }  ,fill opacity=1 ] (98.1,201) .. controls (98.1,199.29) and (96.71,197.9) .. (95,197.9) .. controls (93.29,197.9) and (91.9,199.29) .. (91.9,201) .. controls (91.9,202.71) and (93.29,204.1) .. (95,204.1) .. controls (96.71,204.1) and (98.1,202.71) .. (98.1,201) -- cycle ;
%Shape: Circle [id:dp9912552517460305] 
\draw  [fill={rgb, 255:red, 0; green, 0; blue, 0 }  ,fill opacity=1 ] (140,154) .. controls (140,152.29) and (138.61,150.9) .. (136.9,150.9) .. controls (135.19,150.9) and (133.8,152.29) .. (133.8,154) .. controls (133.8,155.71) and (135.19,157.1) .. (136.9,157.1) .. controls (138.61,157.1) and (140,155.71) .. (140,154) -- cycle ;
%Shape: Circle [id:dp9751888348153464] 
\draw  [fill={rgb, 255:red, 0; green, 0; blue, 0 }  ,fill opacity=1 ] (141,201) .. controls (141,199.29) and (139.61,197.9) .. (137.9,197.9) .. controls (136.19,197.9) and (134.8,199.29) .. (134.8,201) .. controls (134.8,202.71) and (136.19,204.1) .. (137.9,204.1) .. controls (139.61,204.1) and (141,202.71) .. (141,201) -- cycle ;
%Straight Lines [id:da5784051618746391] 
\draw [color={rgb, 255:red, 0; green, 0; blue, 0 }  ,draw opacity=1 ][line width=1.5]    (316,108) -- (279,108) ;
\draw [shift={(275,108)}, rotate = 360] [fill={rgb, 255:red, 0; green, 0; blue, 0 }  ,fill opacity=1 ][line width=0.08]  [draw opacity=0] (6.97,-3.35) -- (0,0) -- (6.97,3.35) -- cycle    ;
%Straight Lines [id:da42540886491325036] 
\draw [color={rgb, 255:red, 0; green, 0; blue, 0 }  ,draw opacity=1 ][line width=1.5]    (275,112) -- (275,148) ;
\draw [shift={(275,152)}, rotate = 270] [fill={rgb, 255:red, 0; green, 0; blue, 0 }  ,fill opacity=1 ][line width=0.08]  [draw opacity=0] (6.97,-3.35) -- (0,0) -- (6.97,3.35) -- cycle    ;
%Straight Lines [id:da7845568881257365] 
\draw [color={rgb, 255:red, 0; green, 0; blue, 0 }  ,draw opacity=1 ][line width=1.5]    (275,160) -- (275,196) ;
\draw [shift={(275,200)}, rotate = 270] [fill={rgb, 255:red, 0; green, 0; blue, 0 }  ,fill opacity=1 ][line width=0.08]  [draw opacity=0] (6.97,-3.35) -- (0,0) -- (6.97,3.35) -- cycle    ;
%Straight Lines [id:da13693052865451105] 
\draw [color={rgb, 255:red, 0; green, 0; blue, 0 }  ,draw opacity=1 ][line width=1.5]    (275,202) -- (311,202) ;
\draw [shift={(315,202)}, rotate = 180] [fill={rgb, 255:red, 0; green, 0; blue, 0 }  ,fill opacity=1 ][line width=0.08]  [draw opacity=0] (6.97,-3.35) -- (0,0) -- (6.97,3.35) -- cycle    ;
%Straight Lines [id:da7276897810457694] 
\draw [color={rgb, 255:red, 0; green, 0; blue, 0 }  ,draw opacity=1 ][fill={rgb, 255:red, 186; green, 204; blue, 217 }  ,fill opacity=1 ][line width=1.5]    (321,202) -- (357,202) ;
\draw [shift={(361,202)}, rotate = 180] [fill={rgb, 255:red, 0; green, 0; blue, 0 }  ,fill opacity=1 ][line width=0.08]  [draw opacity=0] (6.97,-3.35) -- (0,0) -- (6.97,3.35) -- cycle    ;
%Straight Lines [id:da7417816094260516] 
\draw [line width=1.5]    (274,155) -- (310,155) ;
\draw [shift={(314,155)}, rotate = 180] [fill={rgb, 255:red, 0; green, 0; blue, 0 }  ][line width=0.08]  [draw opacity=0] (6.97,-3.35) -- (0,0) -- (6.97,3.35) -- cycle    ;
%Straight Lines [id:da1304861487789093] 
\draw [line width=1.5]    (320,155) -- (356,155) ;
\draw [shift={(360,155)}, rotate = 180] [fill={rgb, 255:red, 0; green, 0; blue, 0 }  ][line width=0.08]  [draw opacity=0] (6.97,-3.35) -- (0,0) -- (6.97,3.35) -- cycle    ;
%Shape: Circle [id:dp7425863903927594] 
\draw  [color={rgb, 255:red, 0; green, 0; blue, 0 }  ,draw opacity=1 ][fill={rgb, 255:red, 0; green, 0; blue, 0 }  ,fill opacity=1 ] (278.1,108.1) .. controls (278.1,106.39) and (276.71,105) .. (275,105) .. controls (273.29,105) and (271.9,106.39) .. (271.9,108.1) .. controls (271.9,109.81) and (273.29,111.2) .. (275,111.2) .. controls (276.71,111.2) and (278.1,109.81) .. (278.1,108.1) -- cycle ;
%Shape: Circle [id:dp5271192908051863] 
\draw  [color={rgb, 255:red, 0; green, 0; blue, 0 }  ,draw opacity=1 ][fill={rgb, 255:red, 0; green, 0; blue, 0 }  ,fill opacity=1 ] (278.1,155.1) .. controls (278.1,153.39) and (276.71,152) .. (275,152) .. controls (273.29,152) and (271.9,153.39) .. (271.9,155.1) .. controls (271.9,156.81) and (273.29,158.2) .. (275,158.2) .. controls (276.71,158.2) and (278.1,156.81) .. (278.1,155.1) -- cycle ;
%Shape: Circle [id:dp11910619148474955] 
\draw  [color={rgb, 255:red, 0; green, 0; blue, 0 }  ,draw opacity=1 ][fill={rgb, 255:red, 0; green, 0; blue, 0 }  ,fill opacity=1 ] (278.1,202) .. controls (278.1,200.29) and (276.71,198.9) .. (275,198.9) .. controls (273.29,198.9) and (271.9,200.29) .. (271.9,202) .. controls (271.9,203.71) and (273.29,205.1) .. (275,205.1) .. controls (276.71,205.1) and (278.1,203.71) .. (278.1,202) -- cycle ;
%Shape: Circle [id:dp49887322226027164] 
\draw  [fill={rgb, 255:red, 0; green, 0; blue, 0 }  ,fill opacity=1 ] (320,155) .. controls (320,153.29) and (318.61,151.9) .. (316.9,151.9) .. controls (315.19,151.9) and (313.8,153.29) .. (313.8,155) .. controls (313.8,156.71) and (315.19,158.1) .. (316.9,158.1) .. controls (318.61,158.1) and (320,156.71) .. (320,155) -- cycle ;
%Shape: Circle [id:dp26146508899689436] 
\draw  [color={rgb, 255:red, 0; green, 0; blue, 0 }  ,draw opacity=1 ][fill={rgb, 255:red, 0; green, 0; blue, 0 }  ,fill opacity=1 ] (321,202) .. controls (321,200.29) and (319.61,198.9) .. (317.9,198.9) .. controls (316.19,198.9) and (314.8,200.29) .. (314.8,202) .. controls (314.8,203.71) and (316.19,205.1) .. (317.9,205.1) .. controls (319.61,205.1) and (321,203.71) .. (321,202) -- cycle ;
%Shape: Circle [id:dp9410087741605067] 
\draw  [fill={rgb, 255:red, 0; green, 0; blue, 0 }  ,fill opacity=1 ] (186.2,154) .. controls (186.2,152.29) and (184.81,150.9) .. (183.1,150.9) .. controls (181.39,150.9) and (180,152.29) .. (180,154) .. controls (180,155.71) and (181.39,157.1) .. (183.1,157.1) .. controls (184.81,157.1) and (186.2,155.71) .. (186.2,154) -- cycle ;
%Shape: Circle [id:dp2132290159364696] 
\draw  [fill={rgb, 255:red, 0; green, 0; blue, 0 }  ,fill opacity=1 ] (187.2,201) .. controls (187.2,199.29) and (185.81,197.9) .. (184.1,197.9) .. controls (182.39,197.9) and (181,199.29) .. (181,201) .. controls (181,202.71) and (182.39,204.1) .. (184.1,204.1) .. controls (185.81,204.1) and (187.2,202.71) .. (187.2,201) -- cycle ;
%Shape: Circle [id:dp21966530941841] 
\draw  [fill={rgb, 255:red, 0; green, 0; blue, 0 }  ,fill opacity=1 ] (366.2,155) .. controls (366.2,153.29) and (364.81,151.9) .. (363.1,151.9) .. controls (361.39,151.9) and (360,153.29) .. (360,155) .. controls (360,156.71) and (361.39,158.1) .. (363.1,158.1) .. controls (364.81,158.1) and (366.2,156.71) .. (366.2,155) -- cycle ;
%Shape: Circle [id:dp7208480015896099] 
\draw  [color={rgb, 255:red, 0; green, 0; blue, 0 }  ,draw opacity=1 ][fill={rgb, 255:red, 0; green, 0; blue, 0 }  ,fill opacity=1 ] (367.2,202) .. controls (367.2,200.29) and (365.81,198.9) .. (364.1,198.9) .. controls (362.39,198.9) and (361,200.29) .. (361,202) .. controls (361,203.71) and (362.39,205.1) .. (364.1,205.1) .. controls (365.81,205.1) and (367.2,203.71) .. (367.2,202) -- cycle ;
%Shape: Circle [id:dp6009880300955615] 
\draw  [fill={rgb, 255:red, 0; green, 0; blue, 0 }  ,fill opacity=1 ] (140.2,107) .. controls (140.2,105.29) and (138.81,103.9) .. (137.1,103.9) .. controls (135.39,103.9) and (134,105.29) .. (134,107) .. controls (134,108.71) and (135.39,110.1) .. (137.1,110.1) .. controls (138.81,110.1) and (140.2,108.71) .. (140.2,107) -- cycle ;
%Shape: Circle [id:dp36725822932993313] 
\draw  [color={rgb, 255:red, 0; green, 0; blue, 0 }  ,draw opacity=1 ][fill={rgb, 255:red, 0; green, 0; blue, 0 }  ,fill opacity=1 ] (319.1,108) .. controls (319.1,106.29) and (317.71,104.9) .. (316,104.9) .. controls (314.29,104.9) and (312.9,106.29) .. (312.9,108) .. controls (312.9,109.71) and (314.29,111.1) .. (316,111.1) .. controls (317.71,111.1) and (319.1,109.71) .. (319.1,108) -- cycle ;
%Straight Lines [id:da422855370973813] 
\draw [line width=1.5]    (462,155) -- (498,155) ;
\draw [shift={(502,155)}, rotate = 180] [fill={rgb, 255:red, 0; green, 0; blue, 0 }  ][line width=0.08]  [draw opacity=0] (6.97,-3.35) -- (0,0) -- (6.97,3.35) -- cycle    ;
%Straight Lines [id:da5816985426161752] 
\draw [line width=1.5]    (508,155) -- (544,155) ;
\draw [shift={(548,155)}, rotate = 180] [fill={rgb, 255:red, 0; green, 0; blue, 0 }  ][line width=0.08]  [draw opacity=0] (6.97,-3.35) -- (0,0) -- (6.97,3.35) -- cycle    ;
%Shape: Circle [id:dp2921124335508758] 
\draw  [color={rgb, 255:red, 0; green, 0; blue, 0 }  ,draw opacity=1 ][fill={rgb, 255:red, 0; green, 0; blue, 0 }  ,fill opacity=1 ] (466.1,108.1) .. controls (466.1,106.39) and (464.71,105) .. (463,105) .. controls (461.29,105) and (459.9,106.39) .. (459.9,108.1) .. controls (459.9,109.81) and (461.29,111.2) .. (463,111.2) .. controls (464.71,111.2) and (466.1,109.81) .. (466.1,108.1) -- cycle ;
%Shape: Circle [id:dp11891016381078123] 
\draw  [color={rgb, 255:red, 0; green, 0; blue, 0 }  ,draw opacity=1 ][fill={rgb, 255:red, 0; green, 0; blue, 0 }  ,fill opacity=1 ] (466.1,155.1) .. controls (466.1,153.39) and (464.71,152) .. (463,152) .. controls (461.29,152) and (459.9,153.39) .. (459.9,155.1) .. controls (459.9,156.81) and (461.29,158.2) .. (463,158.2) .. controls (464.71,158.2) and (466.1,156.81) .. (466.1,155.1) -- cycle ;
%Shape: Circle [id:dp8324902661429101] 
\draw  [color={rgb, 255:red, 0; green, 0; blue, 0 }  ,draw opacity=1 ][fill={rgb, 255:red, 0; green, 0; blue, 0 }  ,fill opacity=1 ] (466.1,202) .. controls (466.1,200.29) and (464.71,198.9) .. (463,198.9) .. controls (461.29,198.9) and (459.9,200.29) .. (459.9,202) .. controls (459.9,203.71) and (461.29,205.1) .. (463,205.1) .. controls (464.71,205.1) and (466.1,203.71) .. (466.1,202) -- cycle ;
%Shape: Circle [id:dp026847219883852302] 
\draw  [fill={rgb, 255:red, 0; green, 0; blue, 0 }  ,fill opacity=1 ] (508,155) .. controls (508,153.29) and (506.61,151.9) .. (504.9,151.9) .. controls (503.19,151.9) and (501.8,153.29) .. (501.8,155) .. controls (501.8,156.71) and (503.19,158.1) .. (504.9,158.1) .. controls (506.61,158.1) and (508,156.71) .. (508,155) -- cycle ;
%Shape: Circle [id:dp4411014134217637] 
\draw  [color={rgb, 255:red, 0; green, 0; blue, 0 }  ,draw opacity=1 ][fill={rgb, 255:red, 0; green, 0; blue, 0 }  ,fill opacity=1 ] (509,202) .. controls (509,200.29) and (507.61,198.9) .. (505.9,198.9) .. controls (504.19,198.9) and (502.8,200.29) .. (502.8,202) .. controls (502.8,203.71) and (504.19,205.1) .. (505.9,205.1) .. controls (507.61,205.1) and (509,203.71) .. (509,202) -- cycle ;
%Shape: Circle [id:dp8541933295101316] 
\draw  [fill={rgb, 255:red, 0; green, 0; blue, 0 }  ,fill opacity=1 ] (554.2,155) .. controls (554.2,153.29) and (552.81,151.9) .. (551.1,151.9) .. controls (549.39,151.9) and (548,153.29) .. (548,155) .. controls (548,156.71) and (549.39,158.1) .. (551.1,158.1) .. controls (552.81,158.1) and (554.2,156.71) .. (554.2,155) -- cycle ;
%Shape: Circle [id:dp1040609706606419] 
\draw  [color={rgb, 255:red, 0; green, 0; blue, 0 }  ,draw opacity=1 ][fill={rgb, 255:red, 0; green, 0; blue, 0 }  ,fill opacity=1 ] (555.2,202) .. controls (555.2,200.29) and (553.81,198.9) .. (552.1,198.9) .. controls (550.39,198.9) and (549,200.29) .. (549,202) .. controls (549,203.71) and (550.39,205.1) .. (552.1,205.1) .. controls (553.81,205.1) and (555.2,203.71) .. (555.2,202) -- cycle ;
%Shape: Circle [id:dp019419404196235224] 
\draw  [color={rgb, 255:red, 0; green, 0; blue, 0 }  ,draw opacity=1 ][fill={rgb, 255:red, 0; green, 0; blue, 0 }  ,fill opacity=1 ] (507.1,108) .. controls (507.1,106.29) and (505.71,104.9) .. (504,104.9) .. controls (502.29,104.9) and (500.9,106.29) .. (500.9,108) .. controls (500.9,109.71) and (502.29,111.1) .. (504,111.1) .. controls (505.71,111.1) and (507.1,109.71) .. (507.1,108) -- cycle ;
%Straight Lines [id:da44944940181660686] 
\draw [line width=1.5]    (401,85.25) -- (401,230.75) ;

% Text Node
\draw (360,213) node [anchor=north west][inner sep=0.75pt]   [align=left] {0};
% Text Node
\draw (313,213) node [anchor=north west][inner sep=0.75pt]   [align=left] {1};
% Text Node
\draw (269.7,213) node [anchor=north west][inner sep=0.75pt]   [align=left] {2};
% Text Node
\draw (248,150) node [anchor=north west][inner sep=0.75pt]   [align=left] {3};
% Text Node
\draw (248,102) node [anchor=north west][inner sep=0.75pt]   [align=left] {4};
% Text Node
\draw (334.45,102) node [anchor=north west][inner sep=0.75pt]   [align=left] {5};
% Text Node
\draw (548,213) node [anchor=north west][inner sep=0.75pt]   [align=left] {0};
% Text Node
\draw (501,213) node [anchor=north west][inner sep=0.75pt]   [align=left] {1};
% Text Node
\draw (457.7,213) node [anchor=north west][inner sep=0.75pt]   [align=left] {2};
% Text Node
\draw (436,150) node [anchor=north west][inner sep=0.75pt]   [align=left] {3};
% Text Node
\draw (436,102) node [anchor=north west][inner sep=0.75pt]   [align=left] {4};
% Text Node
\draw (523,102) node [anchor=north west][inner sep=0.75pt]   [align=left] {5};

\end{tikzpicture}

    \caption{Left: An instance of a preference graph. Middle: $\textsc{NextPath}$ highlighted in blue. Right: The enumeration of paths with the edges removed.}
    \label{fig:enter-label}
\end{figure}
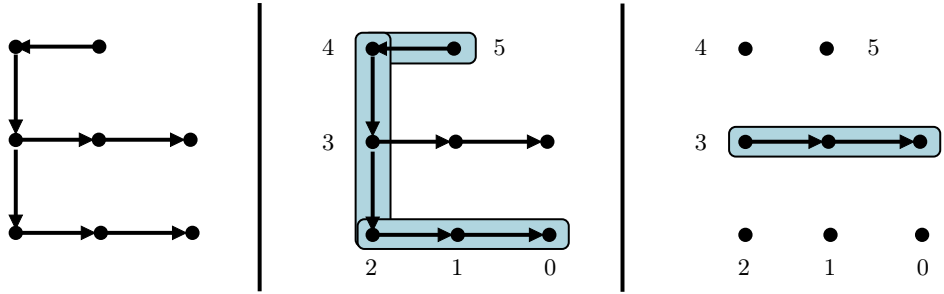

\begin{lemma}
    Given an algorithm's monotone assignment as input, the algorithm \textsc{PathPrice} assigns l(v) equal to the longest path stemming from each vertex $v \in G$. 
\end{lemma}

\begin{proof}

Let $G'$ refer to the preference graph before any edges have been removed by \textsc{PathPrice}. For a given vertex $v$, define $l(v)$ as valid if \textsc{PathPrice} assigns $l(v)$ equal to the length of the longest path stemming from the vertex $v$. 

For each iteration of the while loop on line 7, the following invariant holds: for each vertex v, $l(v)=NaN$, or it is a valid value. Before the loop executes, no vertex has its value set, so the statement is vacuously true. 

By way of induction, for the $k^{th}$ iteration, assume that the invariant has held true up to this point. We will show that it holds true after the $k^{th}$ iteration. $\textsc{NextPath}$ returns the longest path $P_k$ in $G'$ where at least 1 vertex does not have a value and path length is defined as ((\# of vertices $> 1$) + $l(v_k)$). We will induct on the unset vertices (as defined in line 14) in $P_k$, showing that the length assigned to each vertex is equal to the length of the longest path extending from it. 

Last vertex in the path: If this vertex has not been set then it is initialized to zero. There are no edges leading out from it or $P_k$ would not be a maximal path, so the valid length of the last vertex is 0.

Assume that the $j^{th}$ vertex's value has not been set, and all previous $l(v_i)$, where $i$ is from $k$ to $j+1$, set in the path are equal to the length of the longest paths starting from that vertex $v_i$. Let $l(v_{j+1})$ be the value of the $(j+1)^{th}$ vertex in $P_k$. If there is another path leading from the $j^{th}$ vertex whose ( $\#$ of vertices in path + $l(v_k)$) is greater than $l(v_{j+1})$, then $P_k$ would not have been a maximal path, which is a contradiction. Thus $l(v_{j+1}) + 1$ is a valid length for the $j^{th}$ vertex. All other values in the path that have previously been set are covered by the inductive hypothesis, thus every vertex in the path has a valid length.

Thus, for every vertex $v \in G'$ that the \textsc{PathPrice} sets values for, $l(v)$ is the length of the longest path that stems from that vertex.
\end{proof}

Next, we will show that the function $\rho$ output by \textsc{PathPrice} creates a pricing scheme that mimics the given assignment schemes. 

\begin{lemma}
    A greedy agent responding to the prices set by $\textsc{PathPrice}$ on a monotone assignment scheme $A$ chooses precisely the set which would have been chosen for them by $A$.
\end{lemma}

\begin{proof}
    We will show that after $\textsc{PathPrice}$ has run, for every edge in the preference graph $i \rightarrow j$, the price of $S_i$ is greater than $S_j$. Since the edges correspond to the assignment of elements in the intersection of the two sets by a given assignment, having $\pi(S_i) > \pi(S_j)$ ensures that elements introduced in $S_i \cap S_j$ would select $S_j$ over $S_i$, and by doing so replicate the assignment $\mathcal{A}$.

    Since the property defined in Lemma 3 holds after the execution of $\textsc{PathPrice}$, it is the case that for every edge $u \rightarrow v \in G$, $l(u) > l(v)$. Assume for the sake of contradiction there exists an edge $u \rightarrow v \in G$ where $l(u) \leq l(v)$. Let $P_v$ be the length of the longest path that starts at $v$. The length of the path $P'$, defined as the path:\[
    u + P_v
    \] is greater than $l(u)$, because $l(u) \leq l(v)$ and the length of $P'$ is $l(v) + 1$. Thus $P'$ is a path of length greater than $l(u)$ beginning at vertex $u$, which is a contradiction, since lemma 3 tells us that $l(u)$ is the exact length of the longest path that starts from vertex $u$. Thus we must have $l(u) > l(v)$ for all $u \rightarrow v \in G$. 
    
    Because we have $l(u) > l(v)$ for all $u \rightarrow v \in G$, we know that $\pi(S_u) > \pi(S_v)$, where $S_i$ is the set associated with vertex $i\in G$. This is true because \[
    \pi(S_u) = \rho(S_u) + c_{S_u}
    \]and expanding the definition of $\rho(S_u)$ gives us \[
    \pi(S_u) = l(u) + C_{max} - c_{S_u} + c_{S_u}
    \]and since $l(u) > l(v)$, \[
    l(u) + C_{max} - c_{S_u} + c_{S_u} >l(v) + C_{max} - c_{S_v} + c_{S_v} = \pi(S_v)
    \] Thus we have $\pi(S_u) > \pi(S_v)$ for all $u \rightarrow v \in G$. 
    
    Consider $\eta$ that has yet to arrive and is covered by some sets $\{S_1, S_2,..., S_n\}$ and WLOG, $\mathcal{A}$ assigns $\eta$ to $S_1$. Because $\eta$ is assigned to $S_1$ over all other sets, the preference graph from $\mathcal{A}$ will have an edge from each set in $\{ S_2,..., S_n\}$ to $S_1$. After $\textsc{PathPrice}$ has run, $\pi(S_1)$ will be less than the price of every set with an edge going to it, and if $\eta$ arrives next it will choose $S_1$, replicating $\mathcal{A}$. 
    
    Thus with the pricing scheme $\pi$, any $\eta$ covered by more than one set will select the set defined by $\mathcal{A}(\eta)$. Any element covered by a single set can only make one choice so the price defined by $\pi$ is irrelevant to its choice. So in either case, the pricing scheme $\pi$ will perfectly replicate the behavior of any monotone assignment $\mathcal{A}$.
\end{proof}

The above lemmas combine to categorize precisely which algorithms are priceable:

\begin{theorem}
    An algorithm is priceable if and only if it is monotone.
\end{theorem}

\begin{proof}
We know that an algorithm is priceable if and only if there exists a pricing scheme for it. Lemma 1.2 states that every pricing scheme induces a monotone assignment algorithm. Thus, it directly follows that if an algorithm is priceable, it is monotone.

From Lemma 1.3, we know that for any monotone assignment algorithm, \textsc{PathPrice} can generate a pricing scheme that replicates the assignments of the algorithm. Hence, if an algorithm is monotone, it is priceable.

Thus, an algorithm is priceable if and only if it is monotone.
\end{proof}

\section{A strongly competitive monotone algorithm}\label{sec:PDAlg}
Since \textsc{Greedy} performs badly on OSC, we will observe that a primal-dual algorithm from \cite{buchbinder2009} is a \textit{monotone} OSC algorithm; conveniently, it is also strongly competitive in the context of deterministic algorithms with respect to the frequency parameter $f$. Therefore, by applying the transformation \textsc{PathPrice} to obtain a pricing scheme, we have a strongly competitive dynamic pricing algorithm for the DPSC problem.

\begin{definition}[Primal-Dual by Frequency (reproduced from \cite{buchbinder2009})]
\end{definition}

\begin{lstlisting}[mathescape=true, numbers=left]
Maintain a variable $y_\eta$ for each element $\eta$. 
Let $y_S$ denote  $\sum_{\eta\in S}y_\eta$. 
Pick an ordering of the sets $S_1, S_2, ..., S_m \in \mathcal{S}$.

In response to an arriving uncovered element $\eta$: 
    Until there is a set $S\in \mathcal{S}_\eta$ such that $y_S = c_S$, 
        increase $y_\eta$ continuously. 
    Pick the smallest indexed set in $\mathcal{S}_\eta$ such that $y_S = c_S$.
\end{lstlisting}

As shown in \cite{buchbinder2009}, this algorithm happens to be $f$-competitive. We shall now show it is monotone and can therefore be mimicked by a dynamic pricing algorithm. 

\begin{lemma}
    The above algorithm is monotone and therefore priceable.
\end{lemma}

\begin{proof}
Assume for the sake of contradiction that a preference graph constructed from the primal-dual algorithm contains a cycle. Then  there exists a sequence of sets $(S_0, S_1, \dots, S_n)$ such that for each $0\leq i\leq n$ there exists an $\eta_i\in S_i\cap S_{i+1}$\footnote{Letting $S_{n+1} = S_0$} with $\eta_i\rightarrow S_{i+1}$.

In the primal-dual algorithm, elements are assigned to the set that offers the smallest difference between the cost and the dual sum and breaks ties according to the ordering picked in line 3: that is, an element $\eta$ will be matched to the set $S$ that minimizes $c_S - y_S$ (prior to $\eta$'s arrival). Let $ \delta(S_i) = c_{S_i} - y_{S_i}$ prior to the next arrival.

Let $S_j\lessdot S_k$ denote that $S_k$ is picked over $S_j$ by elements in their intersection, i.e. either (a) $\delta(S_k)<\delta(S_j)$ or (b) $\delta(S_k) = \delta(S_j)$ and $k < j$. Since there is a cycle in the preference graph, this implies that:
\[ \delta(S_n) \lessdot \delta(S_0) \lessdot \delta(S_1) \lessdot \ldots \lessdot \delta(S_{n-1}) \lessdot \delta(S_n) \]

This sequence implies that $S_n\lessdot S_n$, which is a contradiction. Thus, the primal-dual algorithm must always create acyclic preference graphs, satisfying the conditions necessary for priceability.
\end{proof}

The following lowerbound comes from \cite{alon2003}, who used it to obtain a lowerbound of $\Omega\left(\frac{\log n \log m}{\log\log n + \log\log m}\right)$ for most reasonable choices inputs to the OSC problem. For our purposes, it gives us a linear lowerbound for deterministic algorithms with respect to the frequency parameter of the input.

\begin{lemma}[Reproduced from \cite{alon2003}]
    No (deterministic) algorithm for the online set cover problem is better than $f$-competitive.
\end{lemma}

\begin{proof}
    Let $\mathcal{X}$ be the binary numbers that count from 1 to $2^k-1$ for some $k$. Let $\mathcal{S}$ contain the sets $\{S_1,...,S_i,...,S_k\}$, where $S_i$ covers all the binary numbers whose $i^{th}$ least significant bit is 1. 

    The adversary first introduces the number $2^k-1$, where every bit in the number is 1. If the algorithm chooses the set $S_i$, the adversary introduces the number where all bits are 1 except for $S_i$. If this process repeats, the adversary can force the algorithm to choose all $k$ sets, while the optimal solution will be the last set $S_j$ that the algorithm choose, because all the binary numbers introduced had a 1 in the $j^{th}$ bit. Since the frequency of the instance is $k$, $2^k -1$ belongs to every set, and no deterministic algorithm can do better than $k$ times worse than \textsc{Opt}, no deterministic algorithm can be better than $f$-competitive.
\end{proof}

\begin{theorem}
    There is an $f$-competitive dynamic pricing algorithm for the online set cover problem and this is optimal for deterministic algorithms.
\end{theorem}
\begin{proof}
The Primal-Dual by frequency algorithm described in definition 7 is a $\Theta(f)$-competitive algorithm for the OSC problem. Lemma 5 tells us that this algorithm is monotone and can therefore be priced. Applying $\textsc{PathPrice}$ to the Primal-Dual algorithm results in a pricing scheme that is $f$-competitive. Lemma 6 tells us that no deterministic algorithm for the OSC problem can be better than $f$-competitive, so this pricing scheme that results from this algorithm is optimal. 
\end{proof}

\section{Conclusion and Further Work}
In this paper we've identified a particular deterministic algorithm that is strongly competitive with respect to deterministic algorithms for the parameter of $f$. In this section, we'll address a few of the questions that we are most interested in seeing resolved. 

As we have identified the optimal answer for deterministic dynamic pricing algorithms, a natural follow-up would be to consider randomized algorithms. The lowerbound from Lemma 6 induces a $\log f$ lowerbound when applying standard techniques against randomized algorithms, however no algorithm is currently known with a $\log f$ competitive ratio for the OSCP. 

In a related sense, it would also be interesting to see competitive dynamic pricing algorithms with respect to other parameters. With respect to the parameters of $n = |X|$ and $m = |\mathcal{S}|$, \cite{alon2003} developed an $O(\log n \log m)$ competitive algorithm, which can somewhat be considered an adaptation of the algorithm defined in Definition 6. This algorithm is both deterministic and nearly optimal, up to a $O(\log\log m + \log\log n)$ factor. However, this algorithm (and many of the other derivative algorithms defined in \cite{buchbinder2009}) are not monotone and therefore not mimickable by any pricing scheme. Candidly, we'll remark that the biggest hurdle in mimicking this algorithm via prices might appear to be the `reset's induced by the standard online guess-and-double technique the algorithm employs, which one might suspect to be navigable using techniques found in \cite{arndt2023}. However, even if supplied with the cost of the optimal solution ahead of time (and thereby subverting the need to use guess-and-double) the algorithm is still not monotone. Therefore it seems new techniques may be necessary to address these parameters. 

Lastly, one of the strengths of the dynamic pricing framework is its lack of necessity for communication between the client and server: in many of the dynamic pricing algorithms presented historically (see section 5 of \cite{cohen2015}, e.g.), the server setting the prices need not interact with the client at all, and need only see the removal of a resource it maintains. Note that this is not the case for the algorithm defined in Definition 6: it needs to know precisely which element $\eta$ is requested in order to increase its relevant variable $y_\eta$. To the best of our knowledge, the lack of this information (the `identity' of the client) has yet to be shown to make any problem definitively harder in any context for dynamic pricing algorithms, and it would be interesting to see if such problems exist. 

%
% ---- Bibliography ----
%
% BibTeX users should specify bibliography style 'splncs04'.
% References will then be sorted and formatted in the correct style.
%
\bibliographystyle{splncs04}
\bibliography{mybibliography}

\begin{thebibliography}{10}
\providecommand{\url}[1]{\texttt{#1}}
\providecommand{\urlprefix}{URL }
\providecommand{\doi}[1]{https://doi.org/#1}

\bibitem{alon2003}
Alon, N., Awerbuch, B., Azar, Y.: The online set cover problem. In: Proceedings
  of the thirty-fifth annual ACM symposium on Theory of computing. pp. 100--105
  (2003)

\bibitem{arndt2023}
Arndt, S., Ascher, J., Pruhs, K.: An $o(\log n)$-competitive posted-price
  algorithm for online matching on the line (2023),
  \url{https://arxiv.org/abs/2310.12394}

\bibitem{bender2020}
Bender, M., Gilbert, J., Krishnan, A., Pruhs, K.: Competitively pricing parking
  in a tree (2020), \url{https://arxiv.org/abs/2007.07294}

\bibitem{bender2021}
Bender, M., Gilbert, J., Pruhs, K.: A poly-log competitive posted-price
  algorithm for online metrical matching on a spider. In: Bampis, E.,
  Pagourtzis, A. (eds.) Fundamentals of Computation Theory. pp. 67--84.
  Springer International Publishing, Cham (2021)

\bibitem{buchbinder2009}
Buchbinder, N., Naor, J.S., et~al.: The design of competitive online algorithms
  via a primal--dual approach. Foundations and Trends{\textregistered} in
  Theoretical Computer Science  \textbf{3}(2--3),  93--263 (2009)

\bibitem{cohen2019}
Cohen, I.R., Eden, A., Fiat, A., Je{\.z}, {\L}.: Dynamic pricing of servers on
  trees. In: Approximation, Randomization, and Combinatorial Optimization.
  Algorithms and Techniques (APPROX/RANDOM 2019). Schloss-Dagstuhl-Leibniz
  Zentrum f{\"u}r Informatik (2019)

\bibitem{cohen2015}
Cohen, I.R., Eden, A., Fiat, A., Łukasz Jeż: Pricing online decisions: Beyond
  auctions (2015), \url{https://arxiv.org/abs/1504.01093}

\bibitem{eden2018}
Eden, A., Feldman, M., Fiat, A., Taub, T.: Prompt scheduling for selfish
  agents. arXiv preprint arXiv:1804.03244  (2018)

\bibitem{feldman2017}
Feldman, M., Fiat, A., Roytman, A.: Makespan minimization via posted prices.
  In: Proceedings of the 2017 ACM Conference on Economics and Computation. pp.
  405--422 (2017)

\bibitem{Gupta17}
Gupta, A., Krishnaswamy, R., Kumar, A., Panigrahi, D.: Online and dynamic
  algorithms for set cover. In: Proceedings of the 49th Annual ACM SIGACT
  Symposium on Theory of Computing. p. 537–550. STOC 2017, Association for
  Computing Machinery, New York, NY, USA (2017). \doi{10.1145/3055399.3055493},
  \url{https://doi.org/10.1145/3055399.3055493}

\bibitem{im2017}
Im, S., Moseley, B., Pruhs, K., Stein, C.: Minimizing maximum flow time on
  related machines via dynamic posted pricing. In: 25th Annual European
  Symposium on Algorithms (ESA 2017). Schloss-Dagstuhl-Leibniz Zentrum f{\"u}r
  Informatik (2017)

\end{thebibliography}
\end{document}